\documentclass[letterpaper, 10 pt, conference]{ieeeconf}  

\IEEEoverridecommandlockouts                              

\usepackage{graphicx,algorithm,algorithmic,amsthm,amsmath,amsfonts,verbatim,mathtools,enumerate,bm,hyperref,cite}
\usepackage{graphicx}
\usepackage{subcaption}
\usepackage{float}
\usepackage{pgfplots}
\usetikzlibrary{matrix}
\usepackage{pgfplots}
\usepgfplotslibrary{fillbetween}
\definecolor{case_1_color}{rgb}{0.47,0.37,0.94}
\definecolor{case_2_color}{rgb}{1,0.38,0}
\definecolor{case_3_color}{rgb}{1,0.69,0}

\allowdisplaybreaks

\DeclareMathOperator{\vech}{vech}
\DeclareMathOperator{\vect}{vec}

\newcommand{\norm}[1]{\left\lVert#1\right\rVert}
\newcommand{\mnorm}[1]{{\left\vert\kern-0.25ex\left\vert\kern-0.25ex\left\vert #1 
    \right\vert\kern-0.25ex\right\vert\kern-0.25ex\right\vert}}

\newtheorem{proposition}{Proposition}

\title{\LARGE \bf Continuous-Time Estimation-Aware Optimal Control of Mobile Sensors for Target Tracking
}

\author{Michael States, Sharveshwaran Umashankar, Minsen Yuan, and Yue Yu%
\thanks{M. States, S. Umashankar, M. Yuan, and Y. Yu are with the Department of Aerospace Engineering and Mechanics, University of Minnesota, Minneapolis, MN 55455, USA ({\tt\small \{state067, umash002, yuan0450, yuey\}@umn.edu}). The authors would like to thank Demoz Gebre-Egziabher for helpful discussions on this work.}%
}

\begin{document}

\maketitle
\thispagestyle{empty}
\pagestyle{empty}

\begin{abstract}
Target tracking is the problem of estimating the state of a target system using measurements collected (oftentimes) by mobile sensors. Trajectory optimization for mobile sensors must account not only for dynamical and operational constraints, but also for uncertainty in estimating the target system's state. We propose an estimation-aware optimal control method for target tracking that accounts for the nonlinear dynamics of both the mobile sensor and target system, the nonlinear covariance dynamics in extended Kalman filtering, and nonlinear measurement models. Furthermore, it ensures constraints on the mobile sensor are satisfied in continuous time (within numerical precision). In numerical simulations, the proposed method achieves estimation-aware target tracking while satisfying the mobile-sensor constraints in continuous time, even under sparse time discretization. 
\end{abstract}

\section{Introduction}

Target tracking is the problem of estimating the state of a dynamical target system using measurements. These measurements are often collected by \emph{mobile sensors}, which are mobile platforms equipped with onboard sensing units (such as aerial vehicles carrying range or bearing sensors). The measurements often depend jointly on the states of the mobile sensor and the target system. As a result, planning the mobile sensor's trajectory must account for not only the sensor's dynamics and operational constraints, but also how measurements collected along the trajectory affect estimation of the target system's state. Target tracking has a wide range of applications including search and rescue, environmental monitoring, surveillance, and site inspection \cite{Hildmann2019review, Senanayake2016search, Furukawa2006recursive, Xing2023autonomous}. 

Existing work on target tracking uses various metrics to quantify how the mobile sensor's trajectory affects target state estimation. One class of metrics involves the observability Gramian, which characterizes how well measurements collected along the mobile sensor's trajectory can distinguish different initial states of the target \cite{Powel2015empirical, Preiss2018simultaneous, Hausman2017observability, Quenzer2014observability}. Variants of the observability Gramian instead focus on the distinguishability of the target's final state \cite{Salaris2019online}. Another direction uses the Fisher information matrix to quantify the amount of information that measurements provide about the target state \cite{Bostrom2018global, Lecadre1999optimizing, Martinez2006optimal, Oshman1999optimization,Li2025recursive}.

A widely used metric is the state-estimation error covariance, which directly quantifies the uncertainty of the state estimate \cite{Prentice2009belief, Zheng2022belief, Liu2017model, Benhamou2024backpropagation, Bostrom2019informative, Zhou2009multi, Cheng2022optimal, Ha2019periodic, Pinto2022multiagent, Choi2010continuous, Haugen2016monitoring, Braaten2026localizing}. Covariance-based metrics are particularly attractive because widely used estimators, like the Kalman filter, provide a covariance that accounts for process noise and measurement noise. Some methods use these metrics to formulate an estimation-aware optimal control problem by jointly considering the dynamics of the mobile sensor and the covariance of the target system. These often feature discrete-time \cite{Liu2017model, Benhamou2024backpropagation, Bostrom2019informative, Zhou2009multi, Braaten2026localizing, Prentice2009belief, Quenzer2014observability} or continuous-time Kalman filters \cite{Cheng2022optimal, Ha2019periodic, Pinto2022multiagent, Choi2010continuous, Haugen2016monitoring}. 

One limitation of existing target-tracking methods is that they only constrain the mobile sensor state in discrete time. Constraints, such as keep-out zones, are often safety critical and must be satisfied at all times, but since informative sensing often requires aggressive maneuvers that drive the state close to constraint boundaries, constraint violations can occur between discretization points. Existing methods impose such constraints either through explicitly discrete-time formulations \cite{Liu2017model, Bostrom2019informative, Zhou2009multi, Benhamou2024backpropagation, Bostrom2018global, Zheng2022belief} or through direct transcription and collocation of continuous-time models \cite{Haugen2016monitoring, Choi2010continuous, Ha2019periodic, Salaris2019online}. Reducing violations between discretization or collocation points generally requires dense time discretization, which increases the size and computational cost of the resulting trajectory optimization problem. 

Addressing this limitation by enforcing continuous-time constraints introduces an additional computational challenge for target tracking. Recent continuous-time optimal control methods ensure constraint satisfaction between discretization points by differentiating the underlying continuous-time dynamics equation \cite{elango2025continuous,elango2026continuous}. In target tracking, this dynamics equation not only introduces nonlinear coupling between the mobile sensor's state and target covariance, but also must preserve the covariance matrix's symmetric structure. Some recent results have developed efficient backpropagation-based methods for differentiating covariance dynamics for discrete-time extended Kalman filtering \cite{Benhamou2024backpropagation}. However, to the best of our knowledge, efficient differentiation formulas of the continuous-time extended Kalman filter covariance and using them to ensure continuous-time constraints in target tracking has received little if any attention. 

We propose an estimation-aware optimal control method that jointly accounts for the nonlinear measurements and mobile sensor dynamics, as well as target system covariance dynamics in extended Kalman filtering. It ensures state and input constraints on the mobile sensor trajectory at all times (within numerical precision). To this end, we first transform constraints on the mobile sensor into auxiliary continuous-time nonlinear dynamics with boundary conditions. Next, we provide an efficient formula that exploits the symmetry of the covariance matrix to differentiate the continuous-time nonlinear covariance dynamics in extended Kalman filtering and support efficient gradient-based optimization. Finally, we demonstrate the proposed method with the augmented Lagrangian algorithm with two examples, LTE transmitter localization, and ground vehicle tracking. Compared with off-the-shelf solvers that only enforce constraints at discretization points, the proposed method ensures continuous-time constraint satisfaction and supports sparse time discretization. 

\section{Estimation-Aware Target Tracking}

We first introduce the mobile sensor, target dynamics, measurement model, and estimation algorithm. We then formulate the estimation-aware optimal control problem. 

\subsection{Mobile Sensor Dynamics}

Consider a mobile sensor whose nonlinear deterministic dynamics are governed by
\begin{equation*}
    \dot{x}(t) = f(x(t), u(t)),
\end{equation*} 
where \(x(t)\in\mathbb{R}^{n_x}\) is the mobile sensor state and \(u(t)\in\mathbb{R}^{n_u}\) is the control input. Let $u(t) \in \mathcal{U}$ define the input set and \(g(x(t)) \leq 0_{n_g}\) define inequality constraints on the mobile sensor's state.

\subsection{Target Dynamics}

Consider a target system whose nonlinear stochastic dynamics are governed by
\begin{equation*}
    \dot{\xi}(t) = \phi(\xi(t)) + w(t),
\end{equation*} 
where \(\xi(t) \in \mathbb{R}^{n_\xi}\) is the target state at time \(t\), and $w(t) \in \mathbb{R}^{n_\xi}$ is the zero-mean Gaussian process noise with covariance \(Q\in\mathbb{S}_{\succeq 0}^{n_\xi}\), where \(\mathbb{S}_{\succeq 0}\) denotes the set of positive semi-definite symmetric matrices. 


\subsection{Mobile Sensor Measurement Model}

The mobile sensor collects measurements defined by the nonlinear stochastic function 
\begin{equation*}
    y(t) = h(x(t), \xi(t)) + v(t), 
\end{equation*}
where $y(t) \in \mathbb{R}^{n_y}$ is the measurement at time \(t\) that depends on the mobile sensor state \(x(t)\) and target state \(\xi(t)\), and $v(t) \in \mathbb{R}^{n_y}$ is the zero-mean Gaussian measurement noise with covariance \(R\in\mathbb{S}_{\succ 0}^{n_y}\), where \(\mathbb{S}_{\succ 0}\) denotes the set of positive definite symmetric matrices.


\subsection{Belief Dynamics in Extended Kalman Filtering}

The continuous-time extended Kalman filter propagates the mean and covariance of the target state estimate according to the differential equations \cite{Simon2006optimal}
\begin{equation}\label{eqn: belief dynamics}
    \begin{aligned}
        \dot{\mu}(t) & = \phi(\mu(t))+ K(x(t))\big(y(t)-h(x(t), \mu(t))\big),\\
       \dot{P}(t) & = A(t)P(t)+P(t)A(t)^\top +Q\\
        &\quad-P(t)C(x(t))^\top R^{-1}C(x(t))P(t),
    \end{aligned}
\end{equation}
where \(\mu(t)\in\mathbb{R}^{n_\xi}\) is the mean of the estimated target state \(\xi(t)\), \(P(t)\in\mathbb{R}^{n_\xi \times n_\xi}\) is the covariance, \(K(x(t)) = P(t)C(x(t))^\top R^{-1}\) is the Kalman gain, and 
\begin{equation}\label{eqn: Jacobians}
\begin{aligned}
    A(t) & = \left.\frac{\partial}{\partial \xi(t)} \phi(\xi(t))\right|_{\xi(t)=\mu(t)},\\
    C(x(t)) & = \left.\frac{\partial}{\partial \xi(t)} h(x(t), \xi(t))\right|_{\xi(t)=\mu(t)}.
\end{aligned}
\end{equation}
Notice that the Kalman gain depends on the mobile sensor state \(x(t)\). Also note that the true stochastic measurement $y(t)$ is unavailable during path planning, so the best prediction of the measurement \(\hat{y}(t) = h(x(t),\mu(t))\) is used instead. 

\subsection{Estimation-Aware Optimal Control Problem}

When planning the mobile sensor's trajectory, future measurements are unavailable. Therefore we adopt the \emph{zero-innovation} assumption, under which the predicted measurement \(y(t)\) equals its expected value \(h(x(t), \mu(t))\). Under this assumption, the mean dynamics in \eqref{eqn: belief dynamics} simplify to
\begin{equation}\label{eqn: mean dynamics}
    \mu(t)=\mu_0+\int_0^t \phi(\mu(\tau)) d\tau,
\end{equation}
where \(\mu_0\) is the initial mean of the target state estimate. To further simplify the covariance dynamics in \eqref{eqn: belief dynamics}, we consider the following half-vectorization transformation that extracts the lower triangular elements of $P(t)$
\begin{equation}
    p(t)=\vech(P(t)), \enskip P(t)=\vech^{-1}(p(t)),
\end{equation}
where $\vech$ stacks the lower triangular entries of a symmetric matrix column-wise (including the diagonal), $\vech^{-1}$ reconstructs the symmetric matrix from these entries, and $p(t) \in \mathbb{R}^{n_p}$ with $n_p = n_\xi(n_\xi + 1)/2$. 
Applying $\vech$ to \eqref{eqn: belief dynamics} yields
\begin{equation}\label{eqn: covariance vech dynamics}
\begin{aligned}
    r(x(t), p(t))&=\vech(A(t)P(t)+P(t)A(t)^\top +Q)\\
        &\quad-\vech(P(t)C(x(t))^\top R^{-1}C(x(t))P(t)),
\end{aligned}  
\end{equation}
where \(A(t)\) and \(C(x(t))\) are computed via \eqref{eqn: Jacobians} and \eqref{eqn: mean dynamics}.

We now use $p(t)$ to formulate the estimation-aware optimal control problem
\begin{equation} \label{eqn: CT optimization problem}
    \begin{array}{ll}
      \underset{x, p, u}{\mbox{min.}} & \int_0^{t_f} \ell(x(t), p(t), u(t)) dt, \\ 
       \mbox{s. t.}  & \dot{x}(t) = f(x(t), u(t)),\\
     & \dot{p}(t)  = r(x(t), p(t)),\\
     & x(0)=x_0,\enskip p(0) = \vech(P_0), \\
     & x(t_f) = x_f, \\
     & g(x(t))\leq 0_{n_g},\enskip u(t)\in\mathcal{U},
    \end{array}
\end{equation}
where \(\ell(x(t),p(t),u(t)) \colon \mathbb{R}^{n_x} \times \mathbb{R}^{n_p} \times \mathbb{R}^{n_u} \to \mathbb{R}\) is the running cost function, $t_f \in \mathbb{R}$ is the fixed final time of the trajectory, $x_0 \in \mathbb{R}^{n_x}$ is the initial mobile sensor state, $x_f \in \mathbb{R}^{n_x}$ is the final mobile sensor state, and $P_0 \in \mathbb{R}^{n_\xi \times n_\xi}$ is the initial target covariance. 
In implementation, the path planner is initialized with the latest target state estimate \(\mu_0\) in \eqref{eqn: mean dynamics}, and target covariance estimate \(P_0\) provided by the state estimator. 
\section{Discretization with Continuous-Time Constraints}

While we could implement \eqref{eqn: CT optimization problem} directly into existing solvers, we propose augmenting the optimal control problem with an accumulated constraint violation term to maintain continuous-time constraint satisfaction. We then discretize and linearize the resulting augmented problem for use in existing optimization algorithms.



\subsection{Discretization of Optimal Control Problem}

To transform the finite-dimensional optimization problem in \eqref{eqn: CT optimization problem} into a finite-dimensional one, we introduce the following augmented state that simultaneously account for mobile sensor dynamics, continuous-time state constraints, and the integral objective function in \eqref{eqn: CT optimization problem} 
\begin{equation*}   
    \zeta(t)=\begin{bmatrix}
    x(t)^\top & p(t)^\top & \alpha(t) & \beta(t)
    \end{bmatrix}^\top,
\end{equation*}
where $\alpha(t) \in \mathbb{R}$ denotes the accumulated constraint violation at time \(t\), and $\beta(t) \in \mathbb{R}$ denotes the integral of the running cost \(\ell\). The augmented state dynamics are governed by
\begin{equation} \label{eqn: CT dynamics}
    \dot{\zeta}(t) = \psi(\zeta(t),u(t)) \coloneqq \begin{bmatrix}
        f(x(t),u(t)) \\
        r(x(t),p(t)) \\
        \frac{1}{2} \norm{\operatorname{max}\{g(x(t)),0_{n_g}\}}_2^2 \\
        \ell(x(t),p(t),u(t))
    \end{bmatrix},
\end{equation}
where $\norm{\cdot}_2$ denotes the $\ell_2$ norm (for vectors).
By imposing the constraint \(\alpha(t_f) = 0\), we ensure that the accumulated constraint violation is zero. 

We introduce the following uniform time discretization
\begin{equation*}
    t_k=(k-1)\frac{t_f}{N-1}, \enskip k = 1, 2, \ldots, N,
\end{equation*}
where $t_k$ is the $k$-th time step of $N$ total steps. We consider zeroth-order-hold inputs such that, for \(k=1, 2, \ldots, N-1\),
\begin{equation*}
    u(t)=u_k, \enskip t \in [t_k, t_{k+1}],
\end{equation*}
for some \(u_k\in\mathbb{R}^{n_u}\).
Furthermore, let \(\zeta_k = \zeta(t_k)\), \(\xi_k = \xi(t_k)\), and \(y_k = y(t_k)\). By considering \(\{\zeta_k\}_{k=1}^N, \{u_k\}_{k=1}^{N-1}\) as optimization variables, we can transform the continuous-time optimal control problem in \eqref{eqn: CT optimization problem} into the following finite-dimensional optimization problem
\begin{equation} \label{eqn: optimization problem}
    \begin{array}{ll}
      \underset{\{\zeta_k\}_{k=1}^N, \{u_k\}_{k=1}^{N-1}}{\mbox{min.}} & \begin{bmatrix}
          0^\top_{n_\zeta - 1} & 1
      \end{bmatrix} \zeta_N, \\
       \; \; \mbox{s. t.}  & \zeta_{k+1} = \Psi_k(\zeta_k,u_k), \enskip k = 1,2, \ldots, N - 1, \\
     & \zeta_1 = \begin{bmatrix}
         x_0 \\ \vech(P_0) \\ 0_2
     \end{bmatrix}, \enskip M \,
     \zeta_N = \begin{bmatrix}
         x_f \\ 0
     \end{bmatrix},\\
     &  u_k\in\mathcal{U},\enskip k=1,2, \ldots, N-1,
    \end{array}
\end{equation}
where $M = \begin{bmatrix}
          I_{n_x} & 0_{n_x\times n_p} & 0_{n_x\times 1} & 0_{n_x\times 1} \\
          0_{1\times n_x} & 0_{1\times n_p} & 1 & 0
      \end{bmatrix}$ extracts $x_k$ and $\alpha_k$ from $\zeta_k$, \(n_\zeta = n_x + n_p + 2\) is the dimension of the augmented state, 
       and 
      $\Psi_k(\zeta_k,u_k) = \zeta(t_{k+1})$ is the solution given by the following initial value problem
\begin{equation} \label{eqn: initial value problem}
\begin{aligned}
    \zeta(t) &= \zeta(t_k) + \int_{t_k}^{t}\psi(\zeta(\sigma), u_k)d\sigma, \enskip t \in [t_k, t_{k+1}], \\
    \zeta(t_k) &= \zeta_k.
\end{aligned}    
\end{equation}

\subsection{Linearization of Covariance Dynamics}

The difficulty in solving the optimization problem~\eqref{eqn: optimization problem} comes from the nonlinear functions \(\{\Psi_k\}_{k=1}^{N-1}\), each of which is defined implicitly by the solution of an initial value problem similar to \eqref{eqn: initial value problem}. 
To enable gradient-based optimization algorithms (such as augmented Lagrangian methods), we need to evaluate the Jacobians of \(\{\Psi_k\}_{k=1}^{N-1}\) efficiently.
To this end, we let
\begin{equation*}
    \begin{aligned}
        \Phi^\zeta(t) &= \frac{\partial \zeta(t)}{\partial \zeta_k}=I_{n_\zeta}+\int_{t_k}^{t}\frac{\partial \psi(\zeta(\sigma),u_k)}{\partial \zeta(\sigma)} \frac{\partial \zeta(\sigma)}{\partial \zeta_k}d\sigma,
        \\
        \Phi^u(t) &= \frac{\partial \zeta(t)}{\partial u_k} 
        \\
        &= \int_{t_k}^{t} \left( \frac{\partial \psi(\zeta(\sigma),u_k)}{\partial \zeta(\sigma)} \frac{\partial \zeta(\sigma)}{\partial u_k} + \frac{\partial \psi(\zeta(\sigma),u_k)}{\partial u_k} \right) d\sigma.
    \end{aligned}
\end{equation*}
where \(\zeta(t)\) satisfies \eqref{eqn: initial value problem}. 
Based on this definition, we can now express the Jacobians of \(\Psi_k\) with respect to \(\zeta_k\) and \(u_k\) as
\begin{equation*}
    \frac{\partial \Psi_k(\zeta_k,u_k)}{\partial \zeta_k}
    = \Phi^\zeta(t_{k+1}), \qquad
    \frac{\partial \Psi_k(\zeta_k,u_k)}{\partial u_k}
    = \Phi^u(t_{k+1}),
\end{equation*}
where \(\Phi^\zeta(t)\) and \(\Phi^u(t)\), for \(t \in [t_k,t_{k+1}]\), are the solutions of the following ordinary differential equations:
\begin{equation} \label{eqn: jacobian initial value problem}
\begin{aligned}
    \dot{\Phi}^\zeta(t) &= 
    \frac{\partial \psi(\zeta(t),u_k)}{\partial \zeta(t)} \Phi^\zeta(t),
    \\
    \dot{\Phi}^u(t) &= 
    \frac{\partial \psi(\zeta(t),u_k)}{\partial \zeta(t)} \Phi^u(t) + \frac{\partial \psi(\zeta(t),u_k)}{\partial u_k},\\
    \Phi^\zeta(t_k) & = I_{n_\zeta},\enskip \Phi^u(t_k) = 0_{n_\zeta \times n_u}.
\end{aligned}
\end{equation}
We compute \(\zeta(t)\), \(\Phi^\zeta(t)\), and \(\Phi^u(t)\) jointly by solving \eqref{eqn: initial value problem} together with \eqref{eqn: jacobian initial value problem}~\cite{elango2025continuous}.



To solve the initial value problem defined in \eqref{eqn: jacobian initial value problem}, we need to differentiate the continuous-time dynamics function \(\psi(\zeta(t), u(t))\) defined in \eqref{eqn: CT dynamics} with respect to both state \(\zeta(t)\) and input \(u(t)\). The challenge in this differentiation comes from the function \(r(x(t), p(t))\) defined in \eqref{eqn: covariance vech dynamics}, which couples matrix half-vectorization and nonlinear coupling between covariance matrix \(P(t)\) and the Jacobian of the measurement function defined in \eqref{eqn: Jacobians}.  

The following proposition provides an analytical formula to differentiate the function \(r(x(t), p(t))\) with respect to its function inputs.

\begin{proposition}\label{prop: Jacobian}
Let \(P(t)=\vech^{-1}(p(t))\), and let
\begin{equation*}
    \begin{aligned}
         K(x(t)) &=P(t) C(x(t))^\top R^{-1},\\
         A_{\text{cl}}(x(t)) &= A(t)-P(t) C(x(t))^\top R^{-1}C(x(t)),
    \end{aligned} 
\end{equation*}
where \(A(t)\) and \(C(x(t))\) are computed according to \eqref{eqn: Jacobians} and \eqref{eqn: mean dynamics} for all \(t\). Let matrices \(D\) and \(E\) be such that \(\vect (S) = D \vech(S),\enskip \vech(S)=E\vect(S)\)
for any \(S\in\mathbb{R}^{n_\xi\times n_\xi}\), where $\vect$ stacks all the entries of a matrix column-wise. In addition, let matrix \(T\) be such that \(T \vect (N)=\vect (N^\top)\) for any \(N \in \mathbb{R}^{n_y\times n_\xi}\). For the function \(r(x(t), p(t))\) defined by \eqref{eqn: covariance vech dynamics},
\begin{equation} \label{eqn: covariance Jacobians}
\begin{aligned}  
    & \frac{\partial r(x(t), p(t))}{\partial p(t)}= E(I_{n_\xi}\otimes A_{\text{cl}}(x(t))+ A_{\text{cl}}(x(t))\otimes I_{n_\xi})D,\\
    &\frac{\partial r(x(t), p(t))}{\partial \vect(C(x(t))}\\
    &= - E \big(P(t)\otimes K(x(t))+ (K(x(t))\otimes P(t)) T\big), 
\end{aligned}  
\end{equation}
where $\otimes$ denotes the Kronecker product. 
\end{proposition}
\begin{proof}
    See Appendix.
\end{proof}
Based on Proposition~\ref{prop: Jacobian}, we can differentiate the function \(r(x(t), p(t))\) with respect to \(x(t)\) and \(p(t)\) via the chain rule. 
\section{Numerical Simulations}

We use the augmented Lagrangian method to solve \eqref{eqn: optimization problem} using \eqref{eqn: covariance Jacobians} for two examples, LTE transmitter localization and ground vehicle tracking.
We then compare our method with a benchmark obtained by solving
\eqref{eqn: CT optimization problem} using \texttt{IPOPT}~\cite{wachter2006implementation} through \texttt{CasADi}~\cite{Andersson2019casadi}.

\subsection{Objective Function}

The running cost function we consider is given by
\begin{equation} \label{eq: runing_obj}
    \ell\left(x(t),p(t),u(t)\right)
    =
    \frac{\omega_p}{t_f \rho_p}\norm{\operatorname{vech}^{-1}(p(t))}_2
    +
    \frac{\omega_u}{t_f \rho_u}\norm{u(t)}_2^2,
\end{equation}
where explicit dependence on $t$ is dropped for clarity, and \(\norm{\cdot}_2\) denotes the Schatten 2-norm (for matrices). In addition,
where \(\omega_p\), \(\omega_u \in [0, 1]\) are normalized weights such that \(\omega_p + \omega_u = 1\), and \(\rho_p\) and \(\rho_u\) are empirically determined normalization terms based on the maximum values of each objective. We consider three weighting schemes: covariance only (\(\omega_p = 1\)), energy only (\(\omega_u = 1\)), and balanced (\(\omega_p = \omega_u = 0.5\)).

\subsection{Example 1: LTE Transmitter Localization}

For the first example, consider a ground vehicle with a receiver that collects pseudorange-rate measurements from a stationary LTE transmitter \cite{Braaten2026localizing}. The goal is to optimally control the ground vehicle (mobile sensor) to obtain the best position and clock drift estimates of the LTE transmitter (target system). We opt for a unicycle model to showcase nonlinear mobile sensor dynamics. 

\subsubsection{Mobile Sensor Dynamics}

The mobile sensor state is given by
\begin{equation*} 
    x(t) = \begin{bmatrix}
        r_\texttt{S}^\top(t) & v_\texttt{S}(t) & \theta_\texttt{S}(t) 
    \end{bmatrix}^\top,
\end{equation*}
where $r_\texttt{S}(t) \in \mathbb{R}^2$ is the position of the mobile sensor, $v_\texttt{S}(t) \in \mathbb{R}$ is the forward velocity, and $\theta_\texttt{S}(t) \in \mathbb{R}$, is the heading angle measured from the \(x\) axis. The control inputs are
\begin{equation*}
    u(t) = \begin{bmatrix}
        \dot{v}_\texttt{S}(t) & \dot{\theta}_\texttt{S}(t)
    \end{bmatrix}^\top, 
\end{equation*}
where $\dot{v}_\texttt{S}(t)$, $\dot{\theta}_\texttt{S}(t) \in \mathbb{R}$ are the target's forward acceleration and angular velocity with respect to a body-fixed frame, respectively. Its dynamics are governed by
\begin{equation*}
    \dot{x}(t) = \begin{bsmallmatrix}
        v_\texttt{S}(t)\operatorname{cos}(\theta_\texttt{S}(t)) \\
        v_\texttt{S}(t)\operatorname{sin}(\theta_\texttt{S}(t)) \\
        \dot{v}_\texttt{S}(t) \\
        \dot{\theta}_\texttt{S}(t) \\
    \end{bsmallmatrix}.
\end{equation*}

The state inequality constraints are given by
\begin{equation*}
    g(x(t)) = \begin{bsmallmatrix}
        r_\texttt{S}(t) - r_\texttt{S}^\texttt{max} \\ 
        r_\texttt{S}^\texttt{min} - r_\texttt{S}(t) \\
        v_\texttt{S}(t) - v_\texttt{S}^\texttt{max} \\
        v_\texttt{S}^\texttt{min} - v_\texttt{S}(t) 
    \end{bsmallmatrix},
\end{equation*}
where $r_\texttt{S}^\texttt{min}$, $r_\texttt{S}^\texttt{max} \in \mathbb{R}^2$ are the element-wise lower and upper bounds on position, and $v_\texttt{S}^\texttt{min}$, $v_\texttt{S}^\texttt{max} \in \mathbb{R}$ are the lower and upper bounds on velocity. The input constraint is given by
\begin{equation*}
    u^\texttt{min} \leq u_k \leq u^\texttt{max}, \enskip k = 1, 2, \ldots, N - 1,
\end{equation*}
where $u^\texttt{min}$, $u^\texttt{max} \in \mathbb{R}^2$ are the element-wise lower and upper bounds on inputs.

\subsubsection{Target System Dynamics}

The target state is given by
\begin{equation*}
    \xi(t)
    =
    \begin{bmatrix}
    r_\texttt{T}^\top(t) & 
    b(t)
    \end{bmatrix}^{\top},
\end{equation*}
where $r_\texttt{T}(t) \in \mathbb{R}^2$ is the target's \(x\) and \(y\) position, and \(b(t)=c(\dot{\delta}_\texttt{S}(t)-\dot{\delta}_\texttt{T}(t))\) is the measurement bias caused by clock-drift. Here, $c$ is the speed of light and $\dot{\delta}_\texttt{S}(t)-\dot{\delta}_\texttt{T}(t)$ is the difference in drift between the receiver and transmitter clocks.  The target dynamics are governed by
\begin{equation*}
    \dot{\xi}(t)
    = w(t).
\end{equation*}
The nominal target state remains constant, while the process noise accounts for uncertainty in the target-state evolution. Note that the process noise of the target position is near-zero.

\subsubsection{Measurement Model}

The mobile sensor obtains a scalar range-rate, or Doppler-like, measurement given by
\begin{equation*}
h(x(t),\xi(t)) =
\begin{bmatrix}
    v_\texttt{S}(t)\operatorname{cos}(\theta_\texttt{S}(t)) \\
    v_\texttt{S}(t)\operatorname{sin}(\theta_\texttt{S}(t))
\end{bmatrix}^\top \frac{r_\texttt{S}(t) - r_\texttt{T}(t)}{\norm{r_\texttt{S}(t) - r_\texttt{T}(t)}_2} + b(t).
\end{equation*}

\subsubsection{Parameters}

We consider 
\(t_f = 10\) and 
\(N = 5\), with initial conditions 
\(x_0 = \begin{bmatrix} 1 & 1 & 0 & 0 \end{bmatrix}^\top\), \(\xi_0 = \begin{bmatrix} 0 & 0 & 0.001 \end{bmatrix}^\top\), and 
\(P_0 = \operatorname{diag}(4, 4, 1\times10^{-4})\), and final condition 
\(x_f = \begin{bmatrix} 0 & 1 & 0 & 0 \end{bmatrix}^\top\). The constraint parameters are  
\(r_\texttt{S}^\texttt{min} = \begin{bmatrix} -0.25 & 0.75 \end{bmatrix}^\top\), 
\(r_\texttt{S}^\texttt{max} = \begin{bmatrix} 1.25 & 1.25 \end{bmatrix}^\top\), 
\(v_\texttt{S}^\texttt{min} = -0.5\),
\(v_\texttt{S}^\texttt{max} = 0.5\), 
\(u^\texttt{min} = \begin{bmatrix} -0.2 & -0.5 \end{bmatrix}^\top\), and 
\(u^\texttt{max} = \begin{bmatrix} 0.2 & 0.5 \end{bmatrix}^\top\). The noise covariances are  
\(Q = \operatorname{diag}(1\times10^{-6}, 1\times10^{-6}, 1.4193\times10^{-4})\), and 
\(R = 5\). The objective normalization terms are 
\(\rho_p = 5\) and 
\(\rho_p = 0.6\).

\subsection{Example 2: Ground Vehicle Tracking}

For the second example, consider an unmanned aerial vehicle (UAV) that collects ranging measurements to a ground vehicle. The goal is to optimally control the UAV (mobile sensor) to obtain the best estimate of the position, velocity, and heading of the ground vehicle (target system). We opt for a unicycle model to showcase nonlinear target dynamics.

\subsubsection{Mobile Sensor Dynamics}

The mobile sensor state is given by 
\begin{equation*}
    x(t) =
    \begin{bmatrix}
        r_\texttt{S}^\top(t) &
        \dot{r}_\texttt{S}^\top(t)
    \end{bmatrix}^{\top},
\end{equation*}
where \(r_\texttt{S}(t) \in \mathbb{R}^2\) is the mobile sensor's position, and \(\dot{r}_\texttt{S}(t) \in \mathbb{R}^2\) is the mobile sensor's velocity. The control inputs are
\begin{equation*}
    u(t) = \ddot{r}_\texttt{S}(t),
\end{equation*}
where \(\ddot{r}_\texttt{S}(t) \in \mathbb{R}^2\) are the mobile sensor's accelerations. Its dynamics are governed by
\begin{equation*}
    \dot{x}(t) = \begin{bsmallmatrix}
        0_{2} & I_{2} \\
        0_{2} & 0_{2}
    \end{bsmallmatrix}x(t) + \begin{bsmallmatrix}
        0_{2} \\ I_{2}
    \end{bsmallmatrix}u(t).
\end{equation*}

The state inequality constraints are given by
\begin{equation*}
    g(x(t)) = \begin{bmatrix}
        \norm{\dot{r}_\texttt{S}(t)}_2 - \dot{r}_\texttt{S}^\texttt{max} \\
        r_\texttt{ST}^\texttt{min} - \norm{r_\texttt{S}(t) - \tilde{r}_\texttt{T}(t)}_2
    \end{bmatrix},
\end{equation*}
where \(\dot{r}_\texttt{S}^\texttt{max} \in \mathbb{R}\) is the maximum velocity of mobile sensor, \(r_\texttt{ST}^\texttt{min} \in \mathbb{R}\) is the minimum allowable distance between the mobile sensor and target, and \(\tilde{r}_\texttt{T}(t) \in \mathbb{R}^2\) is the mean of the estimated target position given by \eqref{eqn: mean dynamics}. The input constraint is given by
\begin{equation*}
    \norm{u_k}_2 \leq u^\texttt{max}, \enskip k = 1, 2, \ldots, N - 1,
\end{equation*}
where $u^\texttt{max} \in \mathbb{R}$ is the input magnitude upper bound.

\subsubsection{Target System Dynamics}

The target state is given by
\begin{equation*} 
    \xi(t) = \begin{bmatrix}
        r_\texttt{T}^\top(t) & v_\texttt{T}(t) & \theta_\texttt{T}(t) & \dot{\theta}_\texttt{T}(t)
    \end{bmatrix}^\top,\\
\end{equation*}
where $r_\texttt{T}(t) \in \mathbb{R}^2$ is the target's position, and $v_\texttt{T}(t)$, $\theta(t)$, $\dot{\theta}(t) \in \mathbb{R}$ are the target's forward velocity, angle from the \(x\) axis, and angular velocity, respectively. Its dynamics are governed by
\begin{equation*}
    \dot{\xi}(t) = \begin{bsmallmatrix}
        v_\texttt{T}(t)\operatorname{cos}(\theta_\texttt{T}(t)) \\
        v_\texttt{T}(t)\operatorname{sin}(\theta_\texttt{T}(t)) \\
        0 \\
        \dot{\theta}_\texttt{T}(t) \\
        0
    \end{bsmallmatrix} + w(t).
\end{equation*}

\subsubsection{Measurement Model}

The mobile sensor obtains a scalar range measurement given by
\begin{equation*}
    h(x(t),\xi(t)) = \norm{r_\texttt{S}(t) - r_\texttt{T}(t)}_2.
\end{equation*}

\subsubsection{Parameters} 

We consider 
\(t_f = 10\) and 
\(N = 5\), with initial conditions 
\(x_0 = \begin{bmatrix} 0 & 0 & 0 & 0 \end{bmatrix}^\top\), 
\(\xi_0 = \begin{bmatrix} 5 & 5 & 1 & 0 & 0.3 \end{bmatrix}^\top\), and 
\(P_0 = \operatorname{diag}(10, 10, 0.5, 2, 0.1)\), and final condition 
\(x_f = x_0\). The constraint parameters are
\(\dot{r}^\texttt{max}_\texttt{S} = 4\), 
\(r_\texttt{ST}^\texttt{min} = 1\), and 
\(u^\texttt{max} = 2\). The noise covariances are  
\(Q = \operatorname{diag}(0.01, 0.01, 0.1, 0.1, 0.1)\), and 
\(R = 5\). The objective normalization terms are 
\(\rho_p = 10\) and 
\(\rho_u = 0.4\).

\subsection{Results}

\begin{figure}[!htp]\captionsetup[subfigure]{font=large}
    \centering
    \begin{tikzpicture}[baseline]
        \pgfplotslegendfromname{sharedlegendLTE}
    \end{tikzpicture}
    \resizebox{\linewidth}{!}{\input{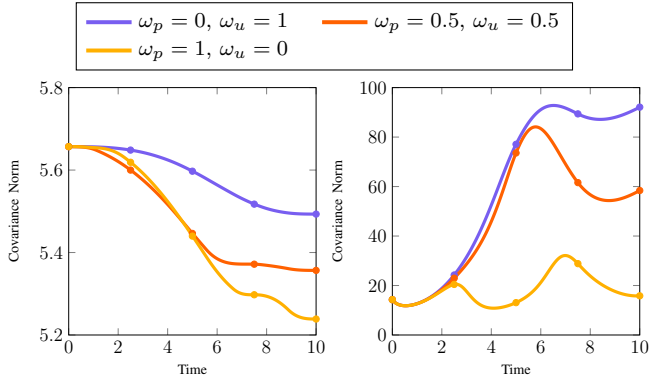}}
    \caption{The trajectories of the error-covariance matrix Schatten 2-norm \(\norm{P(t)}_2\) for LTE transmitter localization (example 1, on the left), and the ground vehicle tracking (example 2, on the right), both solved using the augmented Lagrangian method. We consider three weighting schemes: covariance only (\(\omega_p = 1\)), energy only (\(\omega_u = 1\)), and balanced (\(\omega_p = \omega_u = 0.5\)).
    }
    \label{fig:covariance objective}
\end{figure}

\begin{figure}[!htbp]\captionsetup[subfigure]{font=large}
    \centering
    \resizebox{\linewidth}{!}{\input{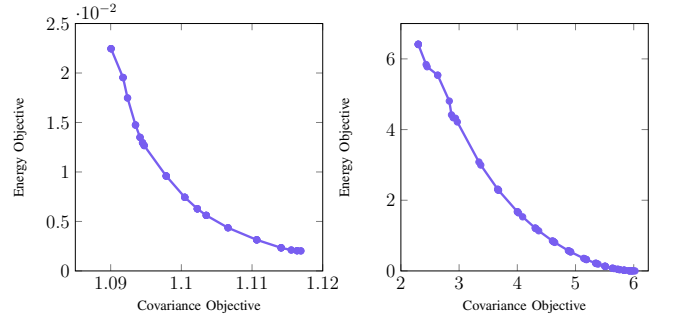}}
    \caption{Pareto fronts showing the trade-offs between the energy and covariance objectives when sweeping across values of the objective weights \(\omega_p\) and \(\omega_u\), for LTE transmitter localization (example 1, on the left), and ground vehicle tracking (example 2, on the right). Both are solved using the Augmented Lagrangian method.
    }
    \label{fig:pareto front}
\end{figure}

\begin{figure}[!htbp]\captionsetup[subfigure]{font=large}
    \centering
    \begin{tikzpicture}[baseline]
        \pgfplotslegendfromname{sharedlegendconstraints}
    \end{tikzpicture}
    \resizebox{\linewidth}{!}{\input{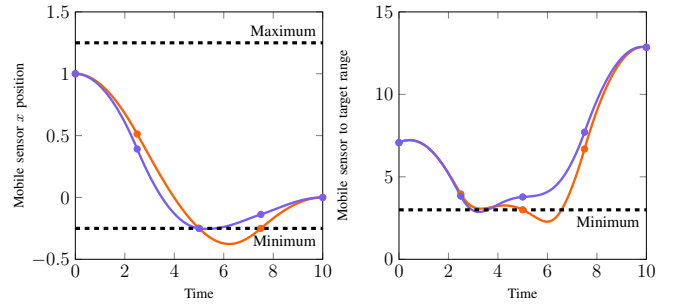}}
    \caption{The $x$ component of \(r_\texttt{s}(t)\) for the LTE transmitter localization example (example 1, on the left), and the range \(\norm{r_\texttt{S}(t) - \tilde{r}_\texttt{T}(t)}_2\) for the ground vehicle tracking example (example 2, on the right). Their respective bounds are indicated by dashed lines. The blue lines indicate our method using the augmented Lagrangian, while the orange lines indicate the benchmark \texttt{CasADi} method. Both examples use the covariance only weighting scheme (\(\omega_p = 1)\).
    }
    \label{fig:constraints}
\end{figure}


Figure~\ref{fig:covariance objective} shows the evolution of \(\norm{P(t)}_2\) for both examples under the three weighting schemes. As the covariance weight \(\omega_p\) increases, the area under the covariance trajectory decreases, indicating improved estimation performance relative to the minimum-energy solution (\(\omega_p\) = 0). 

The Pareto fronts in Fig.~\ref{fig:pareto front} illustrate how more informative measurements require the mobile sensor to perform more aggressive maneuvers. These maneuvers drive the mobile sensor closer to constraint boundaries. 
We further compare the performance of our method with that of the benchmark method in Fig.~\ref{fig:constraints}. The benchmark is obtained by solving a discretized version of \eqref{eqn: CT optimization problem} using \texttt{IPOPT}, a primal-dual interior-point solver~\cite{wachter2006implementation}, through \texttt{CasADi}~\cite{Andersson2019casadi}. As Fig.~\ref{fig:constraints} shows, it exhibits constraint violation between discretization points. While increasing the number of discretization points mitigates this issue, it also increases size and computational cost of the optimization problem. In contrast, our method satisfies the constraints in continuous-time (within numerical precision), even with relatively few discretization points. 

\section{Conclusion}

We presented a method of estimation-aware optimal control that jointly considers the dynamics of the mobile sensor and target covariance, while enforcing constraints in continuous time. We then differentiated the covariance dynamics and exploited the symmetry of the covariance matrix to support efficient optimization. Through two examples, we demonstrated that our method satisfies continuous-time constraints in contrast to benchmark methods.



\section*{Appendix}
\emph{Proof of Proposition~\ref{prop: Jacobian}: }To simplify our notation, we omit any explicit dependence on time \(t\). First, we start with proving the first equation in \eqref{eqn: Jacobians}. Let \(\delta p\) denote a perturbation of vector \(p\). Let \(\delta P=\vech^{-1}(\delta p)\). We can show that \[\vech^{-1}(p+\delta p)=\vech^{-1}(p)+\vech^{-1}(\delta p)=P+\delta P. \]
Hence
\begin{equation}\label{eqn: delta r}
\begin{aligned}
    &r(x, p+\delta p)-r(x, p)\\
    & = -\vech (\delta PC(x)^\top R^{-1}C(x) P+PC(x)^\top R^{-1}C(x)\delta P)\\
    &\quad+\vech(A\delta P+\delta P A)-\vech(\delta P C(x)^\top R^{-1}C(x)\delta P).
\end{aligned}   
\end{equation}
By using the definition of \(\vech\) function, we can show that
\begin{equation}\label{eqn: norm bd}
\begin{aligned}   &\norm{\vech(\delta P C(x)^\top R^{-1}C(x)\delta P)}_2\\
&\leq \norm{\delta P C(x)^\top R^{-1}C(x)\delta P}_F\leq \norm{C(x)}_F^2\norm{R^{-1}}_F\norm{\delta P}_F^2\\
&\leq 2 \norm{C(x)}_F^2\norm{R^{-1}}_F\norm{\delta p}_2^2,
\end{aligned} 
\end{equation}
where the second inequality is due to \cite[Cor. 11.3.7]{bernstein2018scalar}, and the last step holds because each element in \(\delta p\) appears at most twice in the symmetric matrix \(\delta P=\vech^{-1}(\delta p)\). Hence
\begin{equation}\label{eqn: big O del p}
    \vech(\delta P C(x)^\top R^{-1}C(x)\delta P)=O(\norm{\delta p}_2^2),\enskip \text{as } \norm{\delta p}_2\to 0.
\end{equation}

Furthermore, by using \cite[Cor. 9.1.9]{bernstein2018scalar} we can show that
\begin{equation}\label{eqn: vect product}
\begin{aligned}
    &\vect(\delta PC(x)^\top R^{-1}C(x) P+PC(x)^\top R^{-1}C(x)\delta P)\\
    &= (I_{n_\text{T}}\otimes K(x)+K(x)\otimes I_{n_\text{T}})\vect(\delta P),\\
    & \vect(A\delta P+\delta P A)= (I_{n_\text{T}}\otimes A+A\otimes I_{n_\text{T}})\vect(\delta P). 
\end{aligned}  
\end{equation}
where \(K(x)= PC(x)^\top R^{-1}\). By combining the results in \eqref{eqn: vect product}, \eqref{eqn: big O del p}, \eqref{eqn: delta r}, and using the definition of matrix \(E\) and \(D\), we can show that
\begin{equation*}
\begin{aligned}
    &r(x, p+\delta p) - r(x, p)=\\
    &E(I_{n_\text{T}}\otimes A_{\text{cl}}+A_{\text{cl}}\otimes I_{n_\text{T}}) D\vect(\delta P)+ O(\norm{\delta p}_2^2),
\end{aligned}
\end{equation*}
as \(\norm{\delta p}_2\to 0\), where \(A_{\text{cl}} = A-K(x)C\). Therefore, the first equation in \eqref{eqn: Jacobians} holds. 

Next, we prove the second equation in \eqref{eqn: Jacobians}. We first treat \(c=\vect(C(x))\) as an independent variable and compute the Jacobian with respect to \(c\). The desired results then follows by evaluating this Jacobian at \(c=\vect(C(x))\). To simplify the notation, we consider the following function \(c\): \(
q(c)=\vech(AP+PA^\top +Q-PC^\top R^{-1}CP)\), where \(\delta C=\vect^{-1}(c)\). Let \(\delta c\) denote a perturbation of vector \(c\). We can show that
\(\vect^{-1}(c+\delta c)=\vect^{-1}(c)+\vect^{-1}(\delta c)=C+\delta C\).
Hence
\begin{equation}\label{eqn: delta q}
\begin{aligned}
    &q(c+\delta c)-q(c)=-\vech( P (\delta C)^\top R^{-1}\delta C P)\\
    & -\vech (PC^\top R^{-1}\delta C P+P(\delta C)^\top R^{-1}C P).
\end{aligned}   
\end{equation}
By using the definition of the \(\vech\) function, we can show that 
\begin{equation}\label{eqn: norm bd 2}
\begin{aligned}
    &\norm{\vech( P (\delta C)^\top R^{-1}\delta C P)}_2\leq \norm{\delta P (\delta C)^\top R^{-1}\delta C P}_F\\
    &\leq \norm{P}_F^2\norm{R^{-1}}\norm{\delta C}_F^2= \norm{P}_F^2\norm{R^{-1}}\norm{\delta c}_2^2,
\end{aligned}
\end{equation}
where the second inequality is due to \cite[Cor. 11.3.7]{bernstein2018scalar}, and the last step holds because vectorization gives a one-to-one correspondence between the elements of \(\delta C\) and \(\delta c\). Hence
\begin{equation}\label{eqn: big O del c}
    \vech( P (\delta C)^\top R^{-1}\delta C P) = O(\norm{\delta c}_2^2), \enskip \text{as } \norm{\delta c}_2\to 0.
\end{equation}

Furthermore, by using \cite[Cor. 9.1.9]{bernstein2018scalar} and the definition of matrix \(T\), we can show that
\begin{equation}\label{eqn: vec product 2}
\begin{aligned}
    & \vect (PC^\top R^{-1}\delta C P+P(\delta C)^\top R^{-1}C P)\\
    &= (P \otimes PC^\top R^{-1} + (PC^\top R^{-1}\otimes P)T)\delta c,
\end{aligned}
\end{equation}
where we also use the fact that \(\vect(\delta C)=\vect(\vect^{-1} (\delta c))=\delta c\). By combining the results in \eqref{eqn: vec product 2}, \eqref{eqn: big O del c}, \eqref{eqn: delta q}, and using the definition of matrix \(E\), we can show that 
\begin{equation*}
    \begin{aligned}
        &q(c+\delta c)-q(c)\\
        & = -(P \otimes PC^\top R^{-1} + (PC^\top R^{-1}\otimes P)T)\delta c + O(\norm{\delta c}_2^2)
    \end{aligned}
\end{equation*}
as \(\norm{\delta c}_2\to 0\). The second equation in \eqref{eqn: Jacobians} follows by letting \(C=C(x)\) and \(c=\vect(C(x))\) in the equation above.


\bibliographystyle{IEEEtran}
\bibliography{IEEEabrv,reference}

\end{document}